\documentclass[12pt,a4paper]{article}

\usepackage{amssymb,amsthm,latexsym,amsmath,amscd,graphicx,url,xcolor,comment}

\numberwithin{equation}{section}
\theoremstyle{plain}
\newtheorem{thm}{Theorem}[section]
\newtheorem{prop}[thm]{Proposition}
\newtheorem{cor}[thm]{Corollary}
\newtheorem{lemma}[thm]{Lemma}

\theoremstyle{definition}
\newtheorem{defn}[thm]{Definition}
\newtheorem{rmk}[thm]{Remark}
\newtheorem{rmks}[thm]{Remarks}

\newtheorem{quest}[thm]{Question}
\newtheorem{ex}[thm]{Example}

\newcommand{\FF}{\mathbb F}
\newcommand{\A}{\mathcal A}
\newcommand{\C}{\mathcal C}
\newcommand{\D}{\mathcal D}
\newcommand{\MM}{\mathbb M}
\newcommand{\E}{\mathcal E}

\DeclareMathOperator{\rs}{RS}
\DeclareMathOperator{\rk}{rk}
\DeclareMathOperator{\srk}{srk}
\DeclareMathOperator{\supp}{supp}

\DeclareMathOperator{\maxsrk}{maxsrk}
\DeclareMathOperator{\wt}{wt}

\newcommand{\ELG}[1]{{\color{blue} \sf $\heartsuit\heartsuit\heartsuit$ ELG: [#1]}}
\newcommand{\eli}[1]{{\color{red} Elisa G.:#1}}
\newcommand{\umberto}[1]{{\color{olive} Umberto:#1}}
\newcommand{\fla}[1]{{\color{brown} Flavio:#1}}

\title{Integer sequences that are generalized weights of a linear code}
\author{Elisa Gorla, Elisa Lorenzo Garc\'ia, Umberto Mart\'inez-Pe\~{n}as, Flavio Salizzoni}
\date{}

\begin{document}

\maketitle

\begin{abstract}
Which integer sequences are sequences of generalized weights of a linear code? In this paper, we answer this question for linear block codes, rank-metric codes, and more generally for sum-rank metric codes. We do so under an existence assumption for MDS and MSRD codes. We also prove that the same integer sequences appear as sequences of greedy weights of linear block codes, rank-metric codes, and sum-rank metric codes. Finally, we characterize the integer sequences which appear as sequences of relative generalized weights (respectively, relative greedy weights) of linear block codes.
\end{abstract}

\section{Introduction}

Which integer sequences are sequences of generalized weights of a linear code? For linear block codes, this question appears as~\cite[Problem 3.2]{TV}. The question is answered in~\cite{HKY} for binary linear block codes of dimension up to three and in~\cite{Klo} for binary linear block codes of dimension four. In this paper, we fully answer this question for linear block codes, rank-metric codes, and more generally for sum-rank metric codes. 

Necessary conditions for a sequence of positive integers to be the sequence of generalized weights of a linear code appear in~\cite{Wei} for block codes, in~\cite{Rav16} for rank-metric codes, and in~\cite{ourpaper} for the more general situation of sum-rank metric codes. 
The goal of this paper is proving that the necessary conditions from ~\cite{Wei,Rav16,ourpaper} are also sufficient, at least over a field of large enough size. We do so under an existence assumption for MDS and MSRD codes. As this is a strong assumption, but essential throughout the paper, we now discuss its pertinence with respect to the question that we consider in this work. 

The generalized weights of MDS codes can be easily computed using~\cite[Proposition~1]{Wei}. The generalized weights of MSRD codes, on the other side, are computed in~\cite[Theorem~VI.19]{ourpaper}. 
While we know that MRD codes exist for all choices of the parameters and any field size~\cite{Del,gabidulin}, the question of whether MDS or MSRD codes exist is much more complicated and still open for many choices of parameters and field sizes. Since the generalized weights of a code determine its dimension and minimum distance, any code that has the same generalized weights as an MDS or MSRD code is in fact MDS or MSRD. This means in particular that, if we were able to prove that the sequence of integers which agrees with the generalized weights of an MDS or MSRD code with given parameters is the sequence of generalized weights of a linear code over a field of given size, then we would also establish the existence of an MDS or MSRD code with those parameters and over that field. In particular, proving that the necessary conditions from ~\cite{Wei,Rav16,ourpaper} are also sufficient over fields of appropriate size with no further assumptions would establish the MDS Conjecture. In this paper, we choose to assume the existence of chains of nested MDS and MSRD codes in order to prove that the necessary conditions from ~\cite{Wei,Rav16,ourpaper} are also sufficient. 

In addition to characterizing the integer sequences which are the sequence of generalized weights of a linear (block, rank-metric, or sum-rank metric) code, we prove that the same numerical sequences are also sequences of greedy weights of a linear code of the same kind. For the case of linear block codes, we also prove that the same numerical sequences are sequences of relative generalized Hamming weights and of relative greedy weights. Moreover, we discuss the related problem of which integer sequences are the sequence of generalized weights of a linear sum-rank metric subcode of a given MSRD code.

The paper is structured as follows. In Section~\ref{Sec:not} we establish the notation and recall the main definitions that we will use throughout the paper. We also extend the definition of greedy weights and the related notion of chain condition to sum-rank metric codes. In Proposition~\ref{proposition:chaingreedy} we prove that a sum-rank metric code satisfies the chain condition (respectively, the relative chain condition) if and only if its generalized weights and its greedy weights (respectively, its relative generalized weights and relative greedy weights) coincide.

In Section~\ref{sect:hamming} we discuss in detail the case of linear block codes. Theorem~\ref{th all seqs are ghws} shows that a sequence of positive integers is the sequence of generalized weights of a linear block code if and only if it is increasing. Any increasing sequence is also the sequence of greedy weights of a linear block code. In Theorem~\ref{thm:relativeHamming} we prove the analogous result for relative generalized weights and relative greedy weights.

In Section~\ref{sect:sumrank} we discuss the general case of sum-rank metric codes.
Theorem~\ref{thm:genwts_subcodeMSRD} is the main result of this paper. It characterizes the sequences of positive integers which are sequences of generalized weights of a subcode of an MSRD code, under the assumption that a suitable chain of MSRD codes exists. In particular, it characterizes the sequences of positive integers which are sequences of generalized weights of a sum-rank metric code. This important special case is stated as Theorem~\ref{thm:genwts_srk}. In Theorem~\ref{thm:genwts_rk_mn} we state the result for rank-metric codes. Notice that, due to the existence of Gabidulin codes, for rank-metric codes we do not need to assume the existence of a chain of nested MRD codes. 
We conclude the section with Proposition~\ref{existenceMSRD} and Proposition~\ref{existenceMSRDII}, where we construct chains of MSRD codes.

\section{Preliminaries and notation}\label{Sec:not}

Let $r$ be a positive integer and denote by $[r]$ the set $\{1,\ldots,r\}$. For positive integers $m \geq n$ and a prime power $q$, we denote by $\FF_q^{m \times n}$ the set of $m \times n $ matrices with entries in the field $\FF_q$.
We let
$$\MM=\FF_q^{m_1\times n_1}\times\ldots\times\FF_q^{m_t\times n_t},$$
where $t , n_1, \ldots , n_t,m_1, \ldots ,  m_t$ are positive integers such that $m_1\geq  \ldots \geq   m_t$ and $m_i\geq n_i$ for $i\in[t]$. We write $n = n_1 + \ldots + n_{t}$. 
For $C=(C_1,\dots,C_{t})\in \MM$ we define the {\bf sum-rank weight} of $C$ as $$\srk(C)=\sum_{i=1}^{t}\rk(C_i).$$ This weight induces a metric on $\MM$, called the {\bf sum-rank metric}. A {\bf sum-rank metric code} $\C$ is an $\FF_q$-linear subspace of $\MM$ endowed with this metric.
In particular, all codes that we consider in this paper are assumed to be linear.
Notice that if $t = 1$, then $\C\subseteq \FF_q^{m_1 \times n_1}$ is a {\bf rank-metric code}, while if $m_1=m_2=\ldots=m_{t}=1$, then $\MM= \FF_q^{n}$ and $\C\subseteq\FF_q^{n}$ is a {\bf linear block code} endowed with the Hamming metric. In this section, we recall some basic definitions and results on sum-rank metric codes which will be needed in the paper.
We refer the interested reader to~\cite{SRC} for a more comprehensive introduction to the mathematical theory of sum-rank metric codes.

\begin{defn}
The {\bf minimum distance} of a code $0\neq\mathcal{C} \subseteq \MM$ is given by 
$$ d(\C) = \min\{ \mathrm{srk}(C) : C \in \C \setminus 0 \}.$$
The {\bf maximum sum-rank distance} of a code $\C$ is $\maxsrk(\C) = \max\{ \mathrm{srk}(C) : C \in \C  \}$. 
\end{defn}

The next result provides a bound for the dimension of a code in terms of its maximum sum-rank distance.

\begin{thm}[Anticode Bound, {\cite[Theorem III.1]{ourpaper}}]\label{theoremab}
Let $\C\subseteq \MM$ be a sum-rank metric code. Then 
\begin{equation*}
\dim(\C)\leq \max_{C\in\C} \left\{\sum_{i=1}^t m_i \rk(C_i)\right\}.
\end{equation*}
In particular, if $m_1=\ldots=m_t=m$, then $\dim(\C)\leq m\maxsrk(\C).$
\end{thm}

A code that attains the previous bound is called {\bf optimal anticode (OAC)}. Optimal anticodes were classified in~\cite[Theorem III.11]{ourpaper}.
Following~\cite[Definition 2.4]{SRC}, we define the {\bf weight of a code} $\C$ as
$$\wt(\C)=\min\{\maxsrk(\A):\C\subseteq\A=\A_1\times\ldots\times\A_t, \A_i\subseteq\FF_q^{m_i\times n_i} \mbox{ is an OAC for } i\in[t]\}.$$

Notice that if $m_1=\dots=m_{t}=1$, then $\wt(\C)=\left\lvert\{i :C_i\neq0 \text{ for some } C\in\C\}\right\rvert$ is the cardinality of the support of $\C$.

\begin{defn}[{\cite[Definition V.1]{ourpaper}}]\label{definition:genwei}
Let $\C\subseteq \MM$ be a sum-rank metric code. For each $r\in[\dim(\C)]$, the  {\bf $r$-th generalized sum-rank weight} of $\C$ is
\begin{equation*}
d_r(\C) = \min\{\wt(\D):\D\text{ is a subcode of }\C\text{ and }\dim(\D)\geq r\}.
\end{equation*}
\end{defn}

Definition~\ref{definition:genwei} recovers the usual definition of generalized weight from linear block codes from~\cite{HKM,Wei} and the definition of generalized weights for rank-metric codes from~\cite{Rav16}. For rank-metric codes, Definition~\ref{definition:genwei} is different from the definition from~\cite{rgmw}, as discussed in~\cite{G21}.

The next proposition summarizes some of the basic properties of generalized sum-rank weights that we will use throughout the paper.

\begin{prop}[{\cite[Proposition V.6]{ourpaper}}]\label{properties} 
Let $0\neq\C \subseteq\mathcal{D} \subseteq \MM$, then:
\begin{itemize}
\item[(1)] $d_1(\C)= d(\C)$,
\item[(2)] $d_r(\C)\leq d_s(\C)$ for $1\leq r\leq s\leq\dim(\C)$,
\item[(3)] $d_r(\C)\geq d_r(\mathcal{D})$ for $r\in[\dim(\C)]$,
\item[(4)] $d_{\dim(\C)}(\C)\leq n$,
\item[(5)] $d_{r+n_1m_1+\cdots+n_{j-1}m_{j-1}+\delta m_{j}}(\C)\geq d_r(\C)+n_1+\cdots+n_{j-1}+\delta$\\ 
for $j\in[t]$, $r\in[ \dim(\C)-(n_1m_1+\cdots+n_{j-1}m_{j-1}+\delta m_j)]$, and $\delta\in[n_j]$.
\end{itemize}
\end{prop}

Notice that, if $m_1=\cdots=m_t=1$, then $\C$ is a linear block code endowed with the Hamming metric and item (5) implies that $d_{r+1}(\C)>d_r(\C)$ for $r\in[\dim(\C)-1]$.

\begin{defn}
A code $\C\subseteq\MM$ is called {\bf maximum sum-rank distance (MSRD)} if there exist $j\in[t]$ and $0\leq \delta\leq n_j-1$ such that $$d(\C)=\sum_{i=1}^{j-1} n_i+\delta+1 \;\mbox{ and }\; \dim(\C)=\sum_{i=j}^{t} m_in_i-\delta m_{j}.$$
\end{defn}

The Singleton Bound relates the generalized weights and the dimension of a sum-rank metric subcode of $\MM$. By definition, MSRD codes are precisely the sum-rank metric codes that meet both inequalities in the Singleton Bound for $r=1$.

\begin{thm}[Singleton Bound, {\cite[Theorem VI.4]{ourpaper}}]
Let $\C\subseteq\MM$ be a sum-rank metric code and let $r\in[\dim(\C)]$.
Let $j\in[t]$ and $0\leq\delta\leq n_j-1$ be such that
$$d_r(\C)-1\geq\sum_{i=1}^{j-1}n_i+\delta.$$
Then $$\dim(\C)\leq\sum_{i=j}^t m_in_i-m_j\delta+r-1.$$
\end{thm}

It turns out that the generalized weights of an MSRD code are determined by its parameters.

\begin{thm}[{\cite[Theorem VI.19]{ourpaper}}]
Let $j\in[t]$, $0\leq\delta\leq n_j-1$, and let $\C\subseteq\MM$ be an MSRD code of $\dim(\C)=\sum_{i=j}^t m_in_i-\delta m_j$. Define $d_{\max}=\sum_{i=1}^{j-1}n_i+\delta+1$. Let $d_{\max}\leq h\leq n$ and $$r=r_h-r_{d_{\max}-1}-m_k+1,$$ 
where $k=\max\{\nu\mid \sum_{i=1}^{\nu-1}n_i<h\}$. Then  $$d_{r}(\C)=\ldots=d_{r+m_k-1}(\C)=h.$$
\end{thm}

We end the section by discussing some concepts related to generalized weights which appear in the literature on linear block codes. We start by recalling the definition of relative generalized Hamming weights given in~\cite{luo} for linear block codes. 
    
\begin{defn}\label{defn:relativewts}
Let $\C_2 \subsetneq \C_1 \subseteq \FF_q^n$ be nested codes with $ k_j = \dim(\C_j) $, $ j = 1,2 $. The {\bf $r$-th relative generalized Hamming weight} of the pair $\C_2 \subsetneq \C_1 \subseteq \FF_q^n$ is
$$d_r(\C_1, \C_2) = \min \{ |\supp(\D)| : \D \text{ is a subcode of }\C_1, \D \cap \C_2 = 0\text{ and } \dim(\D) \geq r \},$$
for $r\in[k_1 - k_2]$.
\end{defn}
\begin{comment}
Notice that, for rank-metric codes, this definition differs from the definition given in~\cite{rgmw}. However, the definition that one obtains by specializing Definition~\ref{defn:relativewts} to rank-metric codes is the natural one if one adopts the definition of generalized weights from rank-metric codes from~\cite{Rav16}, as we do in this paper.
\end{comment}

The concept of greedy weights was introduced in~\cite{cohen1,cohen2} for linear block codes and then modified to the definition that is generally used today in~\cite{chen99}. We extend the definition to sum-rank metric codes 
in the natural way.

\begin{defn}\label{def:greedySR}
Let $0\neq\C \subseteq\MM $ be a $k$-dimensional code. A {\bf greedy $1$-subcode} is a $1$-dimensional subcode $\D_1\subseteq\C$ such that $\wt(\D_1)=d_1(\C)$. For $2\leq r\leq k$, a {\bf greedy $r$-subcode} is an $r$-dimensional subcode $\D_r\subseteq\C$ of minimum weight among those that contain a greedy $(r-1)$-subcode. For $r\in[k]$, we define the {\bf $r$-th greedy weight of} $ \C $ as
$$g_r(\C)=\wt(\D_r)$$ where $\D_r$ is a greedy $r$-subcode of $\C$.
\begin{comment}
Let $\C_2 \subsetneq \C_1 \subseteq \MM$ be nested codes with $ k_i = \dim(\C_j) $, $ j = 1,2 $. A {\bf relative greedy $1$-subcode} is a $1$-dimensional subcode $\D_1\subseteq\C_1$ such that $\wt(\D_1)=d_1(\C_1,\C_2)$ and $\D_1\cap\C_2=0$.  For $2\leq r\leq k_1-k_2$, a {\bf relative greedy $r$-subcode} is an $r$-dimensional subcode $\D_r\subseteq\C_1$ of minimum weight among those that contain a greedy $(r-1)$-subcode and such that $\D_r\cap\C_2=0$. For $ r\in[k_1 - k_2]$ we define the {\bf $r$-th relative greedy weight} of $ \C $ as
$$g_r(\C_1,\C_2)=\wt(\D_r)$$ where $\D_r$ is a relative greedy $r$-subcode of $\C$.
\end{comment}
\end{defn}

Finally, we extend the concept of chain condition to sum-rank metric codes. 
The chain condition was originally defined in~\cite{wei-product} for linear block codes. 

\begin{defn}
A $ k $-dimensional code $0\neq\C \subseteq \MM$ satisfies the {\bf chain condition} if there exists a chain of subcodes $ 0 \subsetneq \D_1 \subsetneq \D_2 \subsetneq \ldots \subsetneq \D_k = \C $ such that $ d_r(\C) = \wt(\D_r)$ for all $r\in[k]$. 
\begin{comment}
A pair of sum-rank metric codes $\C_2 \subsetneq \C_1 \subseteq \MM $, with $ k_j = \dim(\C_j) $, $ j =1,2 $, satisfies the {\bf relative chain condition} if there exists a chain of subcodes $ 0 \subsetneq \D_1 \subsetneq \D_2 \subsetneq \ldots \subsetneq \D_{k_1-k_2} \subseteq \C_1 $ such that $ \D_r \cap \C_2 = 0$ and $ d_r(\C_1, \C_2) = \wt(\D_r) $ for all $ r\in[k_1 - k_2] $.
\end{comment}
\end{defn}

The next simple proposition clarifies the relation between greedy weights and chain condition. 

\begin{prop}\label{proposition:chaingreedy}
Let $ \C \subseteq \MM$ be a $ k $-dimensional code. 
Then $\C$ satisfies the chain condition if and only if $ d_r (\C) = g_r(\C) $ for $r\in[k]$.   
\end{prop}

\begin{proof}
If $d_r (\C) = g_r(\C)$ for $ r\in[k] $, then let $0\subsetneq\D_1\subsetneq\ldots\subsetneq\D_{k}\subseteq\C$ be a chain of greedy subcodes. This chain satisfies the chain condition, since $d_r(\C)=g_r(\C)=\wt(\D_r)$ for $r\in[k]$. Conversely, if the chain condition holds, then there exists a chain $0\subsetneq\D_1\subsetneq\ldots\subsetneq\D_{k}\subseteq\C$ with $d_r(\C)=\wt(\D_r)$ for $r\in[k]$. One easily proves by induction on~$r$ that each $\D_r$ is a greedy $r$-subcode, since $\D_r$ is an $r$-dimensional subcode of $\C$ with $\wt(\D_r)=d_r(\C)\leq g_r(\C)$. 
\end{proof}

\section{Generalized Hamming weights of linear block codes}\label{sect:hamming}

In this section we focus on linear block codes and generalized Hamming weights. In the notation of Section~\ref{Sec:not}, this corresponds to letting $m_1=1$ and $n_1=\ldots=n_t=1$. Notice that this implies that $m_2=\ldots=m_t=1$ and $t=n$. We start by showing that every increasing sequence of positive integers is the sequence of generalized weights of a linear block code. Then we establish similar results for relative and greedy weights. 

\begin{thm}\label{th all seqs are ghws}
Any increasing sequence of positive integers is the sequence of generalized Hamming weights of a linear block code. In addition, there exists one such code in $\FF_q^n$, provided that $n$ is greater than or equal to the last integer in the sequence and $q\geq n$. Moreover, the code may always be chosen such that it satisfies the chain condition. In particular, any increasing sequence of positive integers is the sequence of greedy weights of a linear block code.
\end{thm}

\begin{proof}
Any non-empty increasing sequence $d_1,\ldots,d_k$ of positive integers can be uniquely written as the juxtaposition of $\ell$ maximal subsequences with the property that any two consecutive entries in the same subsequence differ by one. Let $a_i$ be the length of the $i$-th such subsequence. In other words, the sequence $d_1,\ldots,d_k$ has the form 
$$\underbrace{d_1,d_1+1,\ldots,d_1+a_1-1}_{a_1},\underbrace{d_{a_1+1},\ldots,d_{a_1+1}+a_2-1}_{a_2},\ldots,\underbrace{d_{a_1+\cdots+a_{\ell-1}+1},\ldots,d_{a_1+\cdots+a_{\ell-1}+1}+a_\ell-1}_{a_\ell}$$
for some $\ell\geq 1$ and $a_1,\ldots,a_\ell\geq 1$. Then the length of the sequence is $k=a_1+\cdots+a_{\ell}$ and we let $n=d_k$. 

Denote by $\rs(n,h)$ a Reed-Solomon code of dimension $h$ and minimum distance $n-h+1$ in $\FF_q^n$ for some $q\geq n$.
By induction on $\ell\geq 1$, we prove that there exists a code $\C\subseteq\rs(n,n-d_1+1)$ with the properties that its generalized weights are the given sequence of integers and that it has subcodes $0\subsetneq\C_1\subsetneq\ldots\subsetneq\C_k=\C$ such that $\dim(\C_i)=i$ and $\supp(\C_i)=[d_i]$ for all $i\in[k]$.

If $\ell=1$, then $k=a_1$, $n=d_k=d_1+a_1-1$, and the sequence consists of $a_1$ consecutive positive integers: $d_1,d_1+1,\ldots,d_1+a_1-1$. One can let $\C=\rs(n,n-d_1+1)$. A chain of subcodes $\C_1\subseteq\ldots\subseteq\C_k=\C$ such that $\dim(\C_i)=i$ and $\supp(\C_i)=[d_i]$ for all $i\in[a_1]$ is constructed by letting $\C_i$ be the subcode of $\C$ supported on the first $d_i$ entries, for $i\in[a_1]$.

By induction, suppose that we can construct a code $\D\subseteq\rs(n,n-d_{a_1+1}+1)$, whose generalized weights are 
$$\underbrace{d_{a_1+1},\ldots,d_{a_1+1}+a_2-1}_{a_2},\ldots,\underbrace{d_{a_1+\cdots+a_{\ell-1}+1},\ldots,d_{a_1+\cdots+a_{\ell-1}+1}+a_\ell-1}_{a_\ell}.$$
In addition, $\D$ contains a chain of subcodes $\D_1\subseteq\ldots\subseteq\D_{a_2+\cdots+a_\ell}=\D$ with $\dim(\D_i)=i$ and $\supp(\D_i)=[d_i(\D)]=[d_{i+a_1}]$, for all $i\in[a_2+\cdots+a_\ell]$.

Denote by $e_1,\ldots,e_n$ the elements of the standard basis of $\FF_q^n$. Consider nested codes $\rs(n,n-d_1+1)\supseteq\rs(n,n-d_{a_1+1}+1)\supseteq\D$ and let $$\E_i=\rs(n,n-d_1+1)\cap\langle e_1,\ldots,e_{d_1+i-1}\rangle$$ for $i\in[a_1]$.
Regarded as a subcode of $\langle e_1,\ldots,e_{d_1+i-1}\rangle=\FF_q^{d_1+i-1}$, each $\E_i$ is an MDS code of parameters $(d_1+i-1,i,d_1)$. In fact, $\E_i$ has minimum distance $d(\E_i)=d_1$, since $\rs(n,n-d_1+1)$ contains a codeword supported on the first $d_1\leq d_1+i-1$ coordinates. Moreover, $\E_i$ is the subspace of $\rs(n,n-d_1+1)$ obtained by evaluating the polynomials of degree up to $n-d_1$ whose evaluation in the last $n-d_1-i+1$ points is zero. Therefore, $\dim(\E_i)\geq (n-d_1+1)-(n-d_1-i+1)=i$. Since $\E_i$ has length $d_1+i-1$, it is MDS by the Singleton Bound. Let $\E=\E_{a_1}$.
Let $\C=\D+\E\subseteq\rs(n,n-d_1+1)$.
Since $d(\D)=d_{a_1+1}>d_{a_1}=d_1+a_1-1\geq|\supp(\E)|$, then $\D\cap\E=0$ and $$\dim(\C)=\dim(\D)+\dim(\E)=(a_2+\ldots+a_\ell)+a_1=k.$$ 
Since $\E\subseteq\C\subseteq\rs(n,n-d_1+1)$, then 
$$d_1+i-1=d_i(\rs(n,n-d_1+1))\leq d_i(\C)\leq d_i(\E)=d_1+i-1$$ for $i\in[a_1]$. It follows that the first $a_1$ generalized weights of $\C$ agree with the first $a_1$ elements of the integer sequence. 

Since $\C=\D+\E$ and $\D\cap\E=0$, then any subcode of $\C$ of dimension $i+a_1$ with $i\geq 1$ contains a subcode of $\D$ of dimension $i+a_1-\dim(\E)=i$. It follows that $d_{i+a_1}(\C)\geq d_i(\D)=d_{i+a_1}$.
By the induction hypothesis, for all $i\in[a_2+\cdots+a_\ell]$ there exists $\D_i\subseteq\D$ of $\dim(\D_i)=i$, such that $\supp(\D_i)=[d_{i}(\D)]=[d_{i+a_1}]$. Notice that $d_{i+a_1}\geq d_1+i+a_1-1\geq d_1+a_1$,
hence $[d_{i+a_1}]=\supp(\D_i)\supseteq\supp(\E)=[d_1+a_1-1]$. Since $\D_i\cap\E\subseteq\D\cap\E=0$, then $\D_i+\E$ is an $(i+a_1)$-dimensional subspace of $\C$ with $$\supp(\D_i+\E)=\supp(\D_i)=[d_{i+a_1}],$$ 
so $d_{i+a_1}(\C)\leq d_{i+a_1}$. This proves that $d_{i+a_1}(\C)=d_{i+a_1}$ for $i\in[a_2+\cdots+a_{\ell}]$ and concludes the proof that $d_1,\ldots,d_k$ are the generalized weights of $\C$.

We now show that $\C$ contains a chain of codes as claimed. Recall that $\E$ contains a chain of codes $$0\subsetneq\E_1\subsetneq\ldots\subsetneq\E_{a_1}=\E$$ such that $\dim(\E_i)=i$ and $\supp(\E_i)=[d_i]$, for $i\in[a_1]$. Moreover
$$\E\subsetneq\D_1+\E\subsetneq\ldots\subsetneq\D_{a_2+\cdots+a_{\ell}}+\E=\D+\E=\C,$$ with $\dim(\D_i+\E)=i+a_1$ and $\supp(\D_i+\E)=[d_{i+a_1}]$, for $i\in[a_2+\cdots+a_\ell]$. This proves in addition that $\C$ satisfies the chain condition, since $d_i(\C)=d_i(\E)=\wt(\E_i)$ for $i\in[a_1]$ and $d_i(\C)=d_{i-a_1}(\D)=\wt(\D_{i-a_1}+\E)$ for $i\in[k]\setminus[a_1]$.
Therefore, the greedy weights of $\C$ coincide with its generalized weights by Proposition~\ref{proposition:chaingreedy}. 
\end{proof}

We make the following observations regarding the construction of the code $\C$ in the proof of Theorem~\ref{th all seqs are ghws}.

\begin{rmks}\label{rmks:Hamming}
\begin{itemize}
\item[i)] The construction yields a code $\C$ of length $n$ defined over a field of cardinality $q\geq n$ and which is contained in a $\rs(n,n-d_1+1)$, where $d_1$ is the first entry of the sequence and $n$ is the last. Notice that $n$ points of evaluation in $\FF_q$ are needed in order to have a sequence of nested Reed-Solomon codes, as Reed-Solomon codes using evaluation at infinity are not nested. See also Example~\ref{ex:evaatinfty}.
\item[ii)] The length $n$ of the code can be chosen to be larger than the last generalized weight. In this case, the construction yields a degenerate code. 
\end{itemize}
\end{rmks}

While it is clear that classical Reed-Solomon codes are nested, this is not the case in general if one uses infinity as an evaluation point. We next show a concrete example of this phenomenon.

\begin{ex}\label{ex:evaatinfty}
Let $q$ be a prime power and let $\FF_q=\{\alpha_1,\ldots,\alpha_q\}$. The code $\rs(q+1,k)$ is the image of the encoding map $$\begin{array}{rcl}
\FF_q[x]_{<k} & \rightarrow &\FF_q^{q+1}\\
p(x) & \mapsto & (p(\alpha_1),\ldots,p(\alpha_q),p(\infty))
\end{array}$$
where $\FF_q[x]_{<k}$ is the space of univariate polynomials with coefficients in $\FF_q$ of degree smaller than $k$ and $p(\infty)$ is the coefficient of $x^{k-1}$ in $p(x)$. One may have $\rs(q+1,k)\not\subseteq\rs(q+1,k+1)$. E.g., for any $q$ one has $$\rs(q+1,1)=\langle(1,\ldots,1)\rangle\not\subseteq\rs(q+1,2)=\langle(1,\ldots,1,0),(\alpha_1,\ldots,\alpha_q,1)\rangle.$$
\end{ex}

Theorem~\ref{th all seqs are ghws} may fail over a field of small cardinality. This is connected to Remark~\ref{rmks:Hamming} i), since over a field of small cardinality (chains of nested) MDS codes may not exist for all choices of the parameters. The next examples illustrate what can go wrong over~$\FF_2$.

\begin{ex}
The integer sequence $n-1,n$ satisfies the necessary conditions of Proposition~\ref{properties}. However, by~\cite[Theorem~10]{HKY}, it is the sequence of a generalized weights of a two-dimensional subcode of $\FF_2^n$ if and only if $3(n-1)\leq 2n$, that is $n\leq 3$. This is related to the non-existence of a chain of nested MDS codes of dimension one and two in $\FF_2^n$.
\end{ex}

\begin{ex}
The integer sequence $4,5,7$ satisfies the necessary conditions of Proposition~\ref{properties}. However, \cite[Theorem~10]{HKY} implies that it is not the sequence of generalized weights of a three-dimensional subcode of $\FF_2^7$. This is related to the non-existence of a chain of binary nested MDS codes with parameters $(7,1,7)$ and $(7,4,4)$.
\end{ex}

In the language of matroidal ideals, we can reformulate Theorem~\ref{th all seqs are ghws} as follows.

\begin{cor} 
Any increasing sequence of positive integers can be realized as the sequence of initial degrees of the free modules in a graded minimal free resolution of a matroidal ideal.
\end{cor}

\begin{proof}
The thesis follows by combining~\cite[Theorem~2]{JRV} and Theorem~\ref{th all seqs are ghws}.
\end{proof}

The construction of the code $\C$ in the proof of Theorem~\ref{th all seqs are ghws} can be made explicit, by providing an example of a code which satisfies the chain condition and whose generalized weights are equal to any given increasing sequence of positive integers.

\begin{ex}
Let $ 1 \leq k \leq n \leq q $ and $ 1 \leq d_1 < d_2 < \ldots < d_k \leq n $. Let $\alpha_1,\alpha_2,\ldots,\alpha_n\in\mathbb{F}_q$ be distinct elements and consider the polynomials $ f_i = (x-\alpha_{d_i+1}) (x-\alpha_{d_i+2}) \cdots (x-\alpha_n) \in \mathbb{F}_q[x]$, $i\in[k]$. The $ k $-dimensional code $ \mathcal{C} \subseteq \mathbb{F}_q^n $ with generator matrix
$$ G = \left( \begin{array}{cccc}
f_1(\alpha_1) & f_1(\alpha_2) & \ldots & f_1(\alpha_n) \\
f_2(\alpha_1) & f_2(\alpha_2) & \ldots & f_2(\alpha_n) \\
\vdots & \vdots & \ddots & \vdots \\
f_k(\alpha_1) & f_k(\alpha_2) & \ldots & f_k(\alpha_n) 
\end{array} \right) $$
satisfies $ d_r(\mathcal{C}) = d_r $, for $r \in [k]$. Moreover, $d_r(\mathcal{C})$ is realized by the subspace generated by the first $r$ rows of $G$. In particular, $\C$ satisfies the chain condition and $g_r(\C)=d_r$ for $r\in[k]$.
\end{ex}

We now turn our attention to relative weights. The Singleton Bound for relative generalized Hamming weights~\cite[Section~IV]{luo} provides a necessary condition for a sequence of integers to be the relative generalized Hamming weights of a pair of linear block codes.

\begin{lemma}\label{lemma:relwts}
Let $\C_2 \subsetneq \C_1 \subseteq \mathbb{F}_q^n $ be codes with $ k_j = \dim(\C_j) $, $ j = 1,2 $. Then
$$ d_i(\C_1, \C_2) \leq n - k_1 + i$$ for $i\in[k_1-k_2].$ 
\end{lemma}

In particular, the last relative weight satisfies $ d_{k_1 - k_2}(\C_1, \C_2) \leq n - k_2 $. Since monotonicity also holds by~\cite[Proposition~2]{luo}, every sequence $d_1,\ldots,d_{k_1-k_2}$ of relative generalized Hamming weights of a pair of subcodes of $\FF_q^n$ as above must satisfy $$ 1 \leq d_1 < d_2 < \cdots < d_{k_1-k_2} \leq n - k_2.$$

Fix $0\leq k_2<k_1\leq n$ and let $k=k_1-k_2$.
We next show that any sequence of generalized Hamming weights of a $k$-dimensional subcode of $\FF_q^{n-k_2}$ is also the sequence of relative generalized Hamming weights of a pair of nested subcodes of $\FF_q^{n}$ dimension $k_1$ and $k_2$.

\begin{lemma} \label{lemma ghws are rghws}
Let $ 0 \leq k_2 < k_1 \leq n $ be positive integers. Let $ \C \subseteq \FF_q^{n - k_2} $ be a code of dimension $ k_1 - k_2 $. Define $ \C_1 = \C \times \FF_q^{k_2} $ and $ \C_2 = 0\times \FF_q^{k_2}$, where $0$ denotes the zero code in $\FF_q^{n-k_2}$. Then
$$ d_i(\C_1, \C_2) = d_i(\C), $$
for $i\in[k_1 - k_2]$.
\end{lemma}

\begin{proof}
Let $i\in[k_1 - k_2]$. Since $\C_2 \subsetneq \C_1 \subseteq \mathbb{F}_q^n$, then $d_i(\C_1, \C_2)$ is well-defined. We will prove that $ d_i(\C_1,\C_2) \leq d_i(\C) $ and $ d_i(\C) \leq d_i(\C_1,\C_2) $.

First, let $ \D \subseteq \C $ such that $ |\supp(\D)| = d_i(\C) $ and $ \dim(\D) = i $. Define $ \D^\prime = \D \times 0$, where $0$ denotes the zero code in $\FF_q^{k_2}$. Clearly $ \D^\prime \subseteq \C_1 $, $ \D^\prime \cap \C_2 = 0 $, and $ \dim(\D^\prime) = i $. Hence
$$ d_i(\C_1,\C_2) \leq |\supp(\D^\prime)| = |\supp(\D)| = d_i(\C). $$

Next, take $ \D \subseteq \C_1 $ such that $ \D \cap \C_2 = 0 $, $ \dim(\D) = i $, and $ d_i(\C_1,\C_2) = |\supp(\D)| $. Consider the natural projection map $ \pi : \C_1 \longrightarrow \FF_q^{n - k_2} $ onto the first $ n-k_2 $ coordinates, and define $ \D^\prime = \pi(\D) $. Since $ \ker(\pi) = \C_2 $ and $ \D \cap \C_2 = 0 $, then $ \dim(\D^\prime) = \dim(\D) = i $. Moreover, since $ \C_1 = \C \times \FF_q^{k_2} $, we have that $ \D^\prime \subseteq \pi(\C_1) = \C $. Therefore,
\begin{equation*}
d_i(\C) \leq |\supp(\D^\prime)| \leq |\supp(\D)| = d_i(\C_1,\C_2)
\end{equation*}
for $i\in [k_1-k_2]$.
\end{proof}

As a consequence we can characterize the sequences of positive integers that are the sequence of relative generalized Hamming weights of a pair of nested codes.

\begin{thm}\label{thm:relativeHamming}
Any increasing sequence  of positive integers is the sequence of relative generalized Hamming weights of a pair of nested linear block codes. In addition, there exists a pair of nested codes of dimensions $k_1$ and $k_2$ in $\FF_q^n$ with relative generalized weights $d_1,\ldots,d_{k_1-k_2}$
provided that $0\leq k_2<k_1\leq n$ and $q\geq n-k_2\geq d_{k_1-k_2}$.
Moreover, the pair may always be chosen such that it satisfies the relative chain condition. In particular, any increasing sequence of positive integers is the sequence of relative greedy weights
of a pair of nested linear block codes.
\end{thm}

\begin{proof}
The result follows by combining Lemma~\ref{lemma ghws are rghws}, Theorem~\ref{th all seqs are ghws}, and Proposition~\ref{proposition:chaingreedy}.
\end{proof}

\section{Rank-metric and sum-rank-metric codes}\label{sect:sumrank}

In this section, we extend Theorem~\ref{th all seqs are ghws} to sum-rank-metric codes, and discuss in particular the case of rank-metric codes. We start by discussing a special situation in which the result can be easily proved. 
More precisely, we observe that any non-decreasing sequence of positive integers is the sequence of generalized weights of an $\FF_q$-linear rank-metric code for any $q$ and for large enough $m,n$.

\begin{thm}\label{thm:genwts_rk}
Fix a prime power $q$. Any non-decreasing sequence of positive integers is the sequence of generalized weights of a linear rank-metric code.
\end{thm}

\begin{proof}
Let $d_1,\dots,d_k$ be a non-decreasing sequence of positive integers. Consider the space $\FF_q^{m\times n}$, where $m=\sum_{r=1}^k d_r$ and $n=d_k$. We denote by $E_{i,j}$ the matrix in $\FF_q^{m\times n}$ whose entries are equal to zero, except for a one in position $(i,j)$. For $ r \in [k] $, we let
\begin{equation*}
C_r=\sum_{t=1}^{d_r}E_{d_1+\cdots+d_{r-1}+t,t} \in \FF_q^{m\times n}.
\end{equation*}
Let $\C\subseteq\FF_q^{m\times n}$ be the code generated by $C_1,\dots,C_k$. Then, we have that $\dim(\C)=k$ and $d_r(\C)=d_r$ for all $ r \in [k] $. Indeed, consider the optimal anticode $\A_r=\langle E_{i,j}\mid i \in [m], j \in [d_r] \rangle\subseteq\FF_q^{m\times n}$. Then,
\begin{equation*}
\dim(\C\cap\A_r)\geq \dim\langle C_1,\dots,C_r\rangle=r,
\end{equation*}
hence $d_r(\C)\leq d_r$. Since by construction any matrix in $ \C $ of rank strictly smaller than $ d_r $ is contained in $ \langle C_1 , \ldots, C_{r-1} \rangle $, we conclude that $d_r(\C)=d_r$. 
\end{proof}

In the rest of the section, we characterize the integer sequences which are the generalized weights of sum-rank and rank-metric subcodes of a given ambient space. We start by describing the sequence of generalized weights of the ambient space $\MM$. The next result is a direct consequence of the definition of generalized weights.

\begin{lemma}\label{wts_ambientspace}
The sequence of generalized weights of $\MM=\FF_q^{m_1\times n_1}\times\ldots\times\FF_q^{m_{t}\times n_{t}}$ is $$\underbrace{1,\ldots,1}_{m_1},\underbrace{2,\ldots,2}_{m_1},\ldots,\underbrace{n_1,\ldots,n_1}_{m_1},$$ $$\underbrace{n_1+1,\ldots,n_1+1}_{m_2},\underbrace{n_1+2,\ldots,n_1+2}_{m_2},\ldots,\underbrace{n_1+n_2,\ldots,n_1+n_2}_{m_2},$$
$$\vdots$$
$$\underbrace{n_1+\cdots+n_{t-1}+1,\ldots,n_1+\cdots+n_{t-1}+1}_{m_{t}},\ldots,\underbrace{n_1+\cdots+n_{t},\ldots,n_1+\cdots+n_{t}}_{m_{t}}.$$
A sequence of integers is a subsequence of this sequence if and only if it is a non-decreasing sequence in $[n_1+\cdots+n_{t}]$ such that the integer $n_1+\cdots+n_{i-1}+j$ appears at most $m_{i}$ times, for all $ i \in [t] $ and $ j \in [n_i] $. 
\end{lemma}

The next result provides a necessary condition for a numerical sequence to be the sequence of generalized weights of a sum-rank metric code.

\begin{lemma}[{\cite[Lemma VI.8]{ourpaper}}]\label{lemma:bound}
Let $\C\subseteq\MM=\FF_q^{m_1\times n_1}\times\ldots\times\FF_q^{m_{t}\times n_{t}}$ be a code and let $j\in[t]$, $r\in[\dim(\C)-m_j]$. If
$$d_{r+m_j}(\C)>\sum_{i=1}^{j-1}n_i,\text{  then  }d_{r+m_j}(\C)\geq d_r(\C)+1.$$
\end{lemma}

The previous two lemmas allow us to characterize the subsequences of the sequence of generalized weights of $\MM$.

\begin{lemma}\label{lemma:char_subs}
A non-decreasing sequence of positive integers $d_1,\dots,d_k$ with $d_k\leq n=n_1+\dots+n_{t}$ is a subsequence of the sequence of generalized weights of $\MM$ if and only if 
$$d_{r+m_j}>\sum_{i=1}^{j-1}n_i\text{  implies  }d_{r+m_j}\geq d_r+1,$$
for all $ j \in [t] $ and $ r \in [k - m_j] $. %$r+m_j\leq k$.
\end{lemma}

\begin{proof}
If $d_1,\dots,d_{k}$ is a subsequence of the sequence of generalized weights of $\MM=\FF_q^{m_1\times n_1}\times\ldots\times\FF_q^{m_{t}\times n_{t}}$, then for any pair of positive integers $(r,j)$ such that $j\in[t]$ and $r\in[k-m_j]$, there exists an index $h$ such that $d_{h+m_j}(\MM)=d_{r+m_j}$. If $d_{r+m_j}>\sum_{i=1}^{j-1}n_i$, then
\begin{equation*}
    d_{h+m_j}(\MM)=d_{r+m_j}>\sum_{i=1}^{j-1}n_i.
\end{equation*}
Hence by Lemma~\ref{lemma:bound} we obtain
\begin{equation*}
    d_{h+m_j}(\MM)\geq d_h(\MM)+1\geq d_r+1,
\end{equation*}
where the last inequality follows from the fact that $d_1,\dots,d_{k}$ is a subsequence of the sequence of the generalized weights of $\MM$.

Suppose now that for each pair of positive integers $(r,j)$ such that $j\in[t]$ and $r\in[k-m_j]$ we have that 
$$d_{r+m_j}>\sum_{i=1}^{j-1}n_i\text{  implies  }d_{r+m_j}\geq d_r+1.$$
This implies that for any $j\in[t]$ and $\delta\in[n_j]$ we have at most $m_j$ elements in the sequence $d_1,\dots,d_{k}$ that are equal to $\sum_{i=1}^{j-1}n_i+\delta$. 
Together with $d_{k}\leq n_1+\dots+n_t$, we deduce that $d_1,\dots,d_{k}$ is a subsequence of the sequence of generalized weights of $\MM$ by Lemma~\ref{wts_ambientspace}.
\end{proof}

We give now a necessary condition for a sequence of integers to be the sequence of generalized weights of a subcode of an MSRD code. In Theorem~\ref{thm:genwts_subcodeMSRD}, we will show that such a condition is also sufficient, if we assume the existence of a suitable chain of MSRD codes. Notice that the case $ \D = \MM $ is precisely Lemma~\ref{lemma:bound}.

\begin{prop}\label{prop:subseqMSRD}
Let $\mathcal{C}\subseteq\mathcal{D}\subseteq\MM$ and assume that $\mathcal{D}$ is an MSRD code. Then the sequence of generalized weights of $\mathcal{C}$ is a subsequence of the sequence of generalized weights of $\mathcal{D}$.
\end{prop}

\begin{proof}
By~\cite[Theorem VII.14]{ourpaper} the sequence of generalized weights of $\mathcal{D}$ is the subsequence consisting of the last $\dim(\mathcal{D})$ generalized weights of $\MM$ (see also~\cite[Remark VII.15]{ourpaper}).
By Lemma~\ref{lemma:bound} and Lemma~\ref{lemma:char_subs}, the sequence of generalized weights of $\mathcal{C}$ is a subsequence of the sequence of generalized weights of $\MM$, which is described in Lemma~\ref{wts_ambientspace}. The thesis now follows from the previous two facts and the fact that $d_r(\mathcal{C})\geq d_r(\mathcal{D})$ for $r\in[\dim(\C)]$.
\end{proof}

The next example shows that the conclusion of Proposition~\ref{prop:subseqMSRD} does not necessarily hold for an arbitrary code $\mathcal{D}$.

\begin{ex}\label{ex:nonMSRD}
For $q>2$, let $\alpha\in\FF_q\setminus\{0,1\}$ and let $$\mathcal{D}=\langle (1,1,0,0,0,0),(0,0,1,1,1,0),(0,0,0,\alpha,1,1)\rangle\subseteq \FF_q^6.$$ The generalized weights of $\mathcal{D}$ are $2,4,6$. 
However, $$\C=\langle (1,1,0,0,0,0),(0,0,1,1,1,0)\rangle\subseteq\D$$ has generalized weights $2,5$.
\end{ex}

Our main result is a characterization of the integer sequences that are the sequence of generalized weights of a sum-rank-metric code. We prove our result under the assumption that a suitable chain of MSRD codes exists. Such a chain exists for many choices of parameters, some of which are described in Proposition~\ref{existenceMSRD} and Proposition~\ref{existenceMSRDII}. In particular, it exists in the Hamming metric (corresponding to $m_i=n_i=1$ for all $i\in[t]$) whenever $q \geq n$. In fact, if $q\geq n$, then Reed-Solomon codes form such a chain, see also Remark~\ref{rmks:Hamming} ii). Hence the next theorem may be regarded as a generalization of Theorem~\ref{th all seqs are ghws}. 

\begin{thm}\label{thm:genwts_subcodeMSRD}
Let $\D\subseteq\MM$ be an MSRD code of minimum distance $d$ and assume that there exists a chain of MSRD codes $\D=\D_d\supseteq \D_{d+1}\supseteq \ldots\supseteq \D_n$ such that $d(\D_h)=h$, for $d\leq h\leq n$.
A sequence of positive integers $d_1,\ldots,d_k$ is the sequence of generalized sum-rank weights of a code $\C \subseteq \D$ if and only if $k\leq \dim(\D)$ and $d_1,\ldots,d_k$ is a subsequence of the sequence of generalized weights of $\D$. Moreover, the code $ \C $ may always be chosen such that it satisfies the chain condition. In particular, any subsequence of the sequence of generalized sum-rank weights of $\D$ is the sequence of greedy weights of a subcode of $\D$.
\end{thm}

\begin{proof}
Necessity follows from Proposition~\ref{prop:subseqMSRD}. We now prove sufficiency:
A sequence of integers as in the statement of the theorem can be uniquely written as the juxtaposition of $\ell$ maximal constant subsequences of length $a_2-1,a_3-a_2,\dots,a_{\ell+1}-a_{\ell}$ as
\begin{equation}\label{eqn:sequence}
    d_1,\ldots,d_1,d_{a_2},\ldots,d_{a_2},\ldots,d_{a_{\ell}},\ldots,d_{a_{\ell}},
\end{equation}
for some $1=a_1<a_2<\ldots<a_{\ell}\leq k$ and $d_1<d_{a_2}<\ldots<d_{a_\ell}\leq n$, and where we set $ a_{\ell+1} = k+1 $.
By assumption, the sequence (\ref{eqn:sequence}) is a subsequence of the sequence of generalized weights of $\mathcal{D}$. For each $ h \in [\ell] $ there exist $j\in [t]$ and $0\leq\delta\leq n_j-1$ such that $d_{a_h}=\sum_{i=1}^{j-1} n_i +\delta+1$. By Proposition~\ref{prop:subseqMSRD} and Lemma~\ref{wts_ambientspace}, this implies that $a_{h+1}-a_h\leq m_j$. Since $\D_{d_{a_h}}$ has dimension $\sum_{i=j}^t m_in_i-\delta m_j$, then there exists a subcode $\C_h \subseteq \D_{d_{a_h}}$ of dimension $a_{h+1}-a_{h}$ supported on the first $d_{a_h}$ columns. In fact, let $\A$ be the optimal anticode supported on the first $d_{a_h}$ columns. Then $\dim(\A)=\sum_{i=1}^{j-1} n_im_i +(\delta+1)m_j$, hence $\dim(\A\cap\D_{d_{a_h}})=\dim(\A)+\dim(\D_{d_{a_h}})-\dim(\A+\D_{d_{a_h}})\geq m_j$ and one can choose $\C_h\subseteq\A\cap\D_{d_{a_h}}$. Notice that every nonzero element of $\C_h$ has sum-rank $d_{a_h}$. In addition, $\maxsrk(\sum_{i=1}^{h-1}\C_i)\leq d_{a_{h-1}}<d_{a_h}=d(\C_h)$, hence $(\C_1+\cdots+\C_{h-1})\cap\C_h=0$.
    
Let $\C=\C_1+\cdots+\C_{\ell}$. Then $\C\subseteq\D_{d_1}\subseteq\D$ and $$\dim(\C)=\sum_{h=1}^{\ell}\dim(\C_h)=\sum_{h=1}^{\ell}(a_{h+1} - a_h)=k.$$ To show that (\ref{eqn:sequence}) is the sequence of generalized weights of $\C$, we need to prove that $d_{r}(\C)=d_{a_h}$ for every $h\in [\ell]$ and $r$ such that $a_h\leq r<a_{h+1}$. As before, let $ \A $ be the optimal anticode supported on the first $ d_{a_h} $ columns. By construction, $ \C_1 + \cdots + \C_h \subseteq \A \cap \C $, hence $ \dim(\A \cap \C) \geq a_{h+1}-1 \geq r $, thus $ d_{r}(\C) \leq \maxsrk(\A) = d_{a_h} $. Now, let $\A^\prime$ be an optimal anticode such that $d_r(\C)=\wt(\mathcal{U}_r)=\maxsrk(\A^\prime)$, where $\mathcal{U}_r\subseteq\A^\prime\cap\C$ and $\dim(\mathcal{U}_r)\geq r$. Since $\dim(\A^\prime\cap\C)\geq r\geq a_h$ and $\dim(\C_1+\dots+\C_{h-1})=a_{h}-1$, we have that $\A^\prime\cap(\C_h+\dots+\C_{\ell})\neq 0$. It follows that $d_r(\C)=\maxsrk(\A^\prime)\geq d(\C_h+\dots+\C_{\ell})\geq d_{a_h}$, proving that $d_r(\C)=d_{a_h}$.

Finally, to see that $ \C $ satisfies the chain condition, let $\mathcal{U}_0=0$ and let $ \A_h $ be the optimal anticode supported on the first $ d_{a_h} $ columns for $ h \in [\ell] $. For $a_h\leq r<a_{h+1}$, let $ \mathcal{U}_r \subseteq \A_h \cap \C $ be an $r$-dimensional subspace that contains $\mathcal{U}_{r-1}$. Notice that $\wt(\mathcal{U}_r)\leq d_{a_h}$ since $\mathcal{U}_r\subseteq\A_h$ and $\wt(\mathcal{U}_r)\geq d_r(\C)=d_{a_h}$ since $\mathcal{U}_r$ is an $r$-dimensional subcode of $\C$. Therefore $\wt(\mathcal{U}_r)=d_{a_h}=d_r(\C)$ and the chain $0\subsetneq\mathcal{U}_1\subsetneq\mathcal{U}_2\subsetneq\ldots\subsetneq\mathcal{U}_k=\C$ has the required properties. Since $\C$ satisfies the chain condition, then its greedy weights coincide with its generalized weights by Proposition~\ref{proposition:chaingreedy}. 
\end{proof}

Example~\ref{ex:nonMSRD} also shows that Theorem~\ref{thm:genwts_subcodeMSRD} may fail, if the ambient code is not MSRD.

\begin{ex}
For $q>2$, let $\alpha\in\FF_q\setminus\{0,1\}$ and let $$\mathcal{D}=\langle (1,1,0,0,0,0),(0,0,1,1,1,0),(0,0,0,\alpha,1,1)\rangle\subseteq \FF_q^6.$$ The generalized weights of $\mathcal{D}$ are $2,4,6$. 
However, there is no subcode $\C\subseteq\mathcal{D}$ whose generalized weights are equal to $2,4$.
\end{ex}

\begin{rmk}
The assumption on the sequence of nested MSRD codes is necessary for Theorem~\ref{thm:genwts_subcodeMSRD} to hold, in the following sense. Suppose that every MSRD code has the property that every subsequence of its generalized weights is realized by one of its subcodes. Consider an MSRD code $\mathcal{D}=\mathcal{D}_d$ with $d_1(\mathcal{D})=d$ and the subsequence of its generalized weights consisting of all the generalized weights which are different from $d$. Let $\mathcal{D}_{d+1}$ be a subcode of $\mathcal{D}$ which realizes this subsequence. Then $\mathcal{D}_{d+1}$ is MSRD and we consider the subsequence of its generalized weights consisting of all the generalized weights which are different from $d+1$. By assumption, there is a subcode $\mathcal{D}_{d+2}$ of $\mathcal{D}_{d+1}$ which realizes this subsequence of generalized weights. Proceeding in this fashion, we obtain a sequence of nested MSRD codes as in the statement of Theorem~\ref{thm:genwts_subcodeMSRD}.
\end{rmk}

Theorem~\ref{thm:genwts_subcodeMSRD} allows us to characterize the integer sequences which are the sequence of generalized weights of a sum-rank metric code as follows. 

\begin{thm}\label{thm:genwts_srk}
Assume that there exists a chain of MSRD codes $\MM=\D_1\supseteq \D_2\supseteq \ldots\supseteq \D_n$ such that $d(\D_h)=h$ for $ h \in [n] $.
A sequence of positive integers $d_1,\ldots,d_k$ is the sequence of generalized sum-rank weights of a code $\C \subseteq \MM$ if and only if $d_1,\ldots,d_k$ is a non-decreasing sequence in $[n]$ such that the integer $n_1+\cdots+n_{i-1}+j$ appears in it at most $m_{i}$ times, for all $ i \in [t] $ and $ j \in [n_i] $. 
Moreover, the code $ \C $ may always be chosen such that it satisfies the chain condition. In particular, any subsequence of the sequence of generalized sum-rank weights of $\MM$ is the sequence of greedy weights of a sum-rank metric code $\C\subseteq\MM$.
\end{thm}

\begin{rmk}
The assumption that MSRD codes exist for every choice of parameters is necessary for Theorem~\ref{thm:genwts_srk} to hold, in the following sense. If it is true that every subsequence of the sequence of generalized weights of $\MM$ is the sequence of generalized weights of a code $\C\subseteq\MM$, then for $i\in[t]$ and $j\in[n_i]$ consider the subsequence consisting of all generalized weights of $\MM$ which are bigger than or equal to $n_1+\cdots+n_{i-1}+j$. That is the subsequence
$$\underbrace{n_1+\cdots+n_{i-1}+j,\ldots,n_1+\cdots+n_{i-1}+j}_{m_i},\ldots,\underbrace{n_1+\cdots+n_i,\ldots,n_1+\cdots+n_i}_{m_i},$$
$$\underbrace{n_1+\cdots+n_{i}+1,\ldots,n_1+\cdots+n_{i}+1}_{m_{i+1}},\ldots,\underbrace{n_1+\cdots+n_{i+1},\ldots,n_1+\cdots+n_{i+1}}_{m_{i+1}},$$
$$\vdots$$
$$\underbrace{n_1+\cdots+n_{t-1}+1,\ldots,n_1+\cdots+n_{t-1}+1}_{m_{t}},\ldots,\underbrace{n_1+\cdots+n_{t},\ldots,n_1+\cdots+n_{t}}_{m_{t}}.$$
A code $\C\subseteq\MM$ with those generalized weights is an MSRD code with minimum distance $n_1+\cdots+n_{i-1}+j$. This proves that, if the conclusion of Theorem~\ref{thm:genwts_srk} holds, then MSRD codes exists for every possible choice of minimum distance. 
\end{rmk}

In Theorem~\ref{thm:genwts_srk} we assume the existence of a maximal chain of MSRD codes. This is a stronger assumption than just assuming the existence of MSRD codes for every choice of parameters. The next example shows that it is possible that MSRD codes exist for every choice of parameters, but no maximal chain of MSRD codes exists. In such a situation, the example also shows that it is possible for the conclusion of Theorem~\ref{thm:genwts_srk} to hold, that is, every subsequence of the sequence of generalized weights of $\MM$ is the sequence of generalized weights of a code in $\MM$.

\begin{ex}
Let $q=2$, $t=3$, and $n_i=m_i=1$ for $i\in[3]$, that is $\MM=\FF_2^3$.
The only MDS code in $\mathbb{F}_2^3$ with minimum distance $2$ is the even-weight code. The only MDS code in $\FF_2^3$ with minimum distance $3$ is $\langle(1,1,1)\rangle$. Therefore, $\mathbb{F}_2^3$ contains MSD codes with every possible minimum distance, however it does not contain a chain of nested MDS codes with minimum distances $2$ and $3$. While Theorem~\ref{thm:genwts_srk} does not apply in this situation, it is easy to check by direct inspection that for every subsequence of the sequence $1,2,3$ of generalized weights of $\mathbb{F}_2^3$ there exists a code $\C\subseteq\mathbb{F}_2^3$ whose generalized weights coincide with the chosen subsequence.
\end{ex}

Finally, we state Theorem~\ref{thm:genwts_srk} in the generality of rank-metric codes. This provides a characterization of the integer sequences that are the sequence of generalized weights of a rank-metric subcode of $\FF_q^{m\times n}$ for given $q,n,m$. Notice that Gabidulin codes~\cite{Del, gabidulin} form a chain of nested MRD codes of minimum rank distances $1,2,\ldots,n$ for any $m,n$, and $q$ with $n\leq m$. While several other families of MRD codes are known, see e.g.~\cite{sheekey}, the existence of Gabidulin codes suffices for our purposes.

\begin{thm}\label{thm:genwts_rk_mn}
Fix $1\leq n\leq m$ and $q$ a prime power. A non-decreasing sequence of positive integers $d_1,\ldots,d_k$ is the sequence of generalized rank weights of a linear subcode of $\FF_q^{m\times n}$ if and only if $k\leq mn$, $d_k\leq n$, and any constant subsequence has length at most $m$. 
\end{thm}

\begin{rmk}
Fix any prime power $q$. Given a non-decreasing sequence of positive integers $d_1,\dots,d_k$, any $\C\subseteq\FF_q^{m\times n}$ which has those generalized weights must have $n\geq d_k$. Moreover, one can always make $n=d_k$, since the code that we construct in the proof of Theorem~\ref{thm:genwts_subcodeMSRD} is supported on $ d_k $ columns.
Theorem~\ref{thm:genwts_rk} motivates the question of what is the smallest $m$ for which there exists a code $\C\in\FF_q^{m\times d_k}$ such that $d_i(\C)=d_i$ for $i\in[k]$. Theorem~\ref{thm:genwts_rk_mn} implies that the least $m$ for which there exists $\C\subseteq\FF_q^{m\times n}$ with the desired sequence of generalized weights is the minimum between $n$ and the maximum length of a constant subsequence of $d_1,\dots,d_k$.
\end{rmk}

\begin{rmk}
    The analogous problem for $\FF_{q^m}$-linear rank-metric codes has been studied in~\cite{Nellen}. When $n\leq m$, it is shown in~\cite[Corollary 4.5]{Nellen} that every strictly increasing sequence of integers $d_1<\dots<d_k\leq n$ is the sequence of generalized weights of a $k$-dimensional code $\C\subseteq \FF_{q^m}^n$. The proof of this result is essentially the same as that of Theorem~\ref{thm:genwts_rk_mn}: It suffices to note that Gabidulin codes are $\FF_{q^m}$-linear. When instead $n>m$, the problem is more complex: some positive results can be found in~\cite[Section 5]{Nellen}, but at present it remains largely open.
\end{rmk}

\subsection{Existence of MSRD codes}

Our main motivation for discussing the existence of MSRD codes comes from Theorem~\ref{thm:genwts_subcodeMSRD} and Theorem~\ref{thm:genwts_srk}, where we assume the existence of a chain of MSRD codes with given parameters over a field of size $q$.
Some constructions and necessary conditions may be found in the survey~\cite{SRC} and the references therein. 

In this section, we provide two new constructions of MSRD codes. The next proposition relies on the existence of linearized Reed-Solomon codes.

\begin{prop}\label{existenceMSRD}
Let $j=\min\{i\in[t]\mid m_i=m_t\}$ and let $d=\sum_{i=1}^{j-1} n_i+1$.
Suppose that $n_i \leq m_t$ for all $i\in [t]$. 
If $q>t$, then there exist MSRD codes $\mathcal{C}_n \subseteq \ldots \subseteq \mathcal{C}_{d+1} \subseteq \mathcal{C}_d \subseteq \MM$ with $ d(\mathcal{C}_i) = i $, for $d\leq i\leq n=n_1+\cdots+n_t$.
\end{prop}

\begin{proof}
Since $ q > t $ and $ n_i \leq m_j $ for all $ i \in [t] $, then by~\cite[Definition~31 and Theorem~4]{linearizedRS} there exists a chain of linearized Reed-Solomon codes $\mathcal{D}_n\subseteq\ldots\subseteq\mathcal{D}_{d+1}\subseteq\mathcal{D}_d\subseteq\FF_q^{m_j \times n_1}\times \cdots\times\FF_q^{m_j \times n_t}$. The thesis follows from observing that increasing $m_j$ in the positions $i $ with $i<j$ by adding $m_i-m_j$ rows of zeros to each matrix does not affect the property of being MSRD. 
\end{proof}

\begin{comment}
\begin{rmk}
A sequence of integers which verifies the assumptions of Theorem~\ref{thm:genwts_srk} is the sequence of generalized weights of a sum-rank metric code in $\MM$ if $q>t$ and $n_i\leq m_t$ for all $i\in[t]$. In fact, in this case the existence of a chain of nested MSRD codes with the right parameters follows from Proposition~\ref{existenceMSRD}.
\end{rmk}
\end{comment}

The next result complements Proposition~\ref{existenceMSRD}, under the assumption that one of the $m_i$'s is large enough.

\begin{prop}\label{existenceMSRDII}
Let $j=\min\{i\in[t]\mid m_i=m_t\}$ and let $d=\sum_{i=1}^{j-1} n_i + \delta + 1 $ for some $0\leq\delta\leq n_j-1$. Let $h=\max\{i\in[t]\mid n_i>m_t\}$ and assume that $h\in[t]$ and $m_h\geq m_t\left(\sum_{i=j}^{t}n_i-\delta\right)$. If $q>t-h$, then there exist MSRD codes $\mathcal{C}_n\subseteq\ldots\subseteq\mathcal{C}_{d+1}\subseteq\mathcal{C}_d\subseteq\MM$ with $d(\mathcal{C}_i)=i$, for $d\leq i\leq n=n_1+\cdots+n_t$.
\end{prop}

\begin{proof}
Since $n_j\leq m_j$, then $h<j$. By Proposition~\ref{existenceMSRD} there exists a chain of MSRD codes $ \mathcal{D}_n \subseteq \ldots \subseteq \mathcal{D}_{d+1} \subseteq \mathcal{D}_d \subseteq \FF_q^{m_{h+1} \times n_{h+1}} \times \cdots \times \FF_q^{m_t \times n_t}$ with $ d(\mathcal{D}_i) = i - \sum_{\ell=1}^{h} n_\ell $ for $d\leq i\leq n$. 
Moreover, we claim that there exists a code $ \mathcal{D}_0 \subseteq \FF_q^{m_1 \times n_1} \times \cdots \times \FF_q^{m_h \times n_h}$ with minimum distance $\sum_{i=1}^{h} n_i$ and dimension $m_h$. In fact, for all $\ell\in[h]$ there exists an MRD code $\mathcal{U}_\ell\subseteq\FF_q^{m_\ell\times n_\ell}$ with minimum distance $n_\ell$ and dimension $m_\ell$. For any fixed $\ell\in[h]$, let $U_{i\ell}\in\mathcal{U}_\ell$ be linearly independent matrices with $i\in[m_h]$. Then $N_i=(U_{i1},\ldots,U_{i,h})$, $i\in[m_h]$, are a basis of a code $\D_0$ with the required properties. Let $ M_{i,1}, \ldots, M_{i,D_i} $ be a basis of $ \mathcal{D}_i $, where $ D_i = \dim(\mathcal{D}_i) $ and $d\leq i\leq n$. Finally, let $ \mathcal{C}_i \subseteq \MM$ be the code with basis $ (N_1,M_{i,1}),\ldots,(N_{D_i},M_{i,D_i}) $. 
The construction works, since $D_n\leq\ldots\leq D_d=\sum_{\ell=j}^{t}m_\ell n_\ell-m_j\delta\leq m_h$. 
The code $\C_i$ has minimum distance $i$ and dimension $ D_i $, hence it is MSRD, for $d\leq i\leq n$. By construction we have $ \mathcal{C}_n \subseteq \ldots \subseteq \mathcal{C}_{d+1} \subseteq \mathcal{C}_d\subseteq\MM$.
\end{proof}

\begin{comment}
\begin{quest}
Is there an MSRD code of minimum distance $4$ and dimension $4$ in $\FF_q^{3\times 3}\times\FF_q^{2\times 2}$?
\end{quest}

If it exists, then the projection on $\FF_q^{2\times 2}$ is surjective.

Notice that there is an MSRD code of minimum distance $5$ and dimension $4$ in $\FF_q^{4\times 4}\times\FF_q^{2\times 2}$ for any $q$, and one can construct it as follows. Take a basis $M_1,\ldots,M_4$ of an MDR code of dimension $4$ and minimum distance $4$ in $\FF_q^{4\times 4}$. This exists for every $q$. Let $N_1,\ldots,N_4$ a basis of $\FF_q^{2\times 2}$. Then $\{(M_1,N_1),\ldots,(M_4,N_4)\}$ is a basis of the desired MSRD code.

\begin{lemma}There exists a MSRD code in $\FF_q^{2\times 2}\times\FF_q^{2\times 2}$ with generalized weights $3,3,4,4$ for $q$ big enough. 
\end{lemma}
\begin{proof}
Consider the elements $\begin{pmatrix}a & b\\ c & d\end{pmatrix}\times\begin{pmatrix}a+d & b+mc\\ mb+c & a+m^2d\end{pmatrix}$ with $m$ a non-square in $\FF_q$ with $m^2\neq1$. 
Notice that $ad=bc$ and $(a+d)(a+m^2d)-(b+mc)(mb+c)=0$ implies after making $a=\frac{bc}{d}$ that $(md^2-b^2)(md^2-c^2)=0$. This is only possible if $d=b=c=0$ in $\FF_q$, but then $a=0$. So except the zero element, all others have weight at least $3$.
\end{proof}
\end{comment}

More is known if we restrict to linear block codes and the Hamming metric. It is well known that MDS codes exist whenever $q\geq n-1$. 
The MDS Conjecture states that, if $2\leq k\leq q-1$, this sufficient condition is also necessary. The MDS Conjecture was proven in~\cite{Ball2012} by Ball in several situations, including the case when $q$ is prime. See also~\cite{MDSsurvey} for a more recent survey. The problem of characterizing the parameter sets for which MDS or MSRD codes exist is a highly nontrivial one and remains open in general.

\bibliographystyle{plain}

\end{document}